\documentclass[longbibliography,notitlepage,a4paper,twoside,fleqn,leqno,11pt]{revtex4-1}

\usepackage{hyperref}
\usepackage[utf8]{inputenc}
\usepackage{color}
\usepackage{calc}
\usepackage{newlfont}
\usepackage[english]{babel}
\usepackage{amssymb,amsmath,amsfonts,amsthm}
\usepackage{mathrsfs}
\usepackage{dsfont}
\usepackage[all]{xy}
\usepackage{graphicx}
\usepackage{enumerate}
\DeclareMathOperator{\Tr}{Tr}

\renewcommand{\Im}{\mathrm{Im}}

\newcommand{\de}{\mathrm{d}}

\newcommand{\braket}[2]{ \langle #1 , #2 \rangle }  %braket
\newcommand{\ketbra}[2]{\lvert #1\rangle\langle #2\rvert}%ketbra
\newcommand{\norm}[1]{\bigl\lVert #1 \bigr\rVert} %norm
\newcommand{\ass}[1]{\left\lvert #1 \right\rvert} %absolute value
\newcommand{\ide}[1]{\mathrm{d}#1 \,} %de for integrals (already spaced)

\newcommand{\W}{\widetilde{W}}

\newcommand{\C}{\mathds{C}}

\newcommand{\nnorm}[1]{\bigl\lvert\mspace{-1.65mu}\bigl\lVert #1 \bigr\rVert\mspace{-1.65mu}\bigr\rvert}

\theoremstyle{plain}
\newtheorem{thm}{Theorem}

\newtheorem*{assertion*}{Assertion}
\newtheorem{proposition}{Proposition}

\newtheorem{lemma}{Lemma}
\newtheorem*{lemma*}{Lemma}

\newtheorem*{corollary*}{Corollary}
\theoremstyle{definition}
\newtheorem{definition}{Definition}
\newtheorem*{definition*}{Definition}
\newtheorem*{definitions*}{Definitions}

\theoremstyle{remark}

\newtheorem{remark}{Remark}

\newtheorem*{remark*}{Remark}
\newtheorem*{remarks*}{Remarks}
\begin{document}
\date{\today}
\title{Mean field limit of bosonic systems in partially factorized states and their linear combinations.}
\author{Marco Falconi}
\affiliation{Dipartimento di Matematica, Università di Bologna\\Piazza di Porta San Donato 5 - 40126 Bologna, Italy}
\email[]{m.falconi@unibo.it}
\begin{abstract}
We study the mean field limit of one-particle reduced density matrices, for a bosonic system in an initial state with a fixed number of particles, only a fraction of which occupies the same state, and for linear combinations of such states. In the mean field limit, the time-evolved reduced density matrix is proved to converge: in trace norm, towards a rank one projection (on the state solution of Hartree equation) for a single state; in Hilbert-Schmidt norm towards a mixed state, combination of projections on different solutions (corresponding to each initial datum), for states that are a linear superposition.
\end{abstract}
%\subjclass{Primary 81V70; Secondary 35Q55}
%\keywords{Mean Field Limit; Hartree Equation; Many Body Systems}
\maketitle

\section{Introduction.}
\label{sec:introduction}

The mathematics of mean field limit of quantum systems has been widely investigated in the late 35 years. The first to put on a sound mathematical basis the concept of mean field limit of many-boson systems were \citet{MR530915}, developing an idea by \citet{MR0332046}. Actually, in their work they performed the classical limit $h\to 0$, and used the formalism of Fock space: they showed that, in the limit, bounded functions of annihilation and creation operators converge in some sense to bounded functions of the solution of classical equation corresponding to the system (Hartree equation); also, the quantum evolution between $h$-dependent coherent states converges when $h\to 0$ to the evolution of quantum fluctuations around the classical solution. Their result was extended by \citet{MR2530155}: they studied the convergence of reduced density matrices of many-bosons systems in the mean field limit, and also provided a bound on the rate of convergence. To do that they proved the convergence of normal ordered products of time evolved creation and annihilation operators (averaged on initial states that depend suitably on the number of particles). Further improvements were made by \citet{2011JMP} and \citet{2011arXiv1103.0948C}. The Fock space method has also been used to study the mean field limit of other bosonic systems, such as the Nelson model with cut off, or with a different scaling of the potential \citep[see][for further details]{MR2205462,falconi:012303,2012arXiv1208.0373B,jioonlee:14213}.

Another method has been widely used to study mean field limit of quantum systems, and it is based on a hierarchy of equations called BBGKY. This method has been very successful but has some limitations: in particular, due to its abstract argument, it does not give information on the rate of convergence of reduced density matrices \citep[see][for a review of BBGKY methods, and detailed references]{2013arXiv1301.5494G}.
Recently, yet another simple method has been developed by Pickl \citep{MR2657816,MR2821235}: it avoids the technicalities of BBGKY hierarchies and to introduce the formalism of Fock spaces (in particular the use of Weyl operators).

To sum up, we review briefly the results of the works cited above. Let $\varphi$ be a one-particle normalized state; then $\varphi^{\otimes_n}$ is a state with $n$ particles, all in the same state $\varphi$, and $C(\sqrt{n}\varphi)\Omega$ ($C$ is the Weyl operator, $\Omega$ the vacuum of Fock space) a coherent state with an average number of particles $n$. Now let $\Tr_1\rho_{\varphi^{\otimes_n}}(t)$ and $\Tr_1\rho_{C(\sqrt{n}\varphi)\Omega}(t)$ be the corresponding one-particle reduced density matrices evolved in time by quantum dynamics. In the mean field limit $n\to\infty$, the following convergences in trace norm are proved, for suitable bosonic systems (\citet{MR2821235} proves convergence in operator norm):
\begin{equation}\label{eq:8}
  \Tr\ass{\Tr_1\rho_{\varphi^{\otimes_n}}(t)-\ketbra{\varphi_t}{\varphi_t}}\overset{n\to\infty}{\longrightarrow} 0\; ,\;\Tr\ass{\Tr_1\rho_{C(\sqrt{n}\varphi)\Omega}(t)-\ketbra{\varphi_t}{\varphi_t}}\overset{n\to\infty}{\longrightarrow} 0
\end{equation}
where $\varphi_t$ is the solution of the classical equation with initial datum $\varphi$. A bound on the rate of convergence of order $n^{1/2}$ \citep[e.g.][]{MR2530155,MR2657816,2012arXiv1208.0373B} is given; for some systems, using quantum fluctuations, it can be improved to order $n^{-1}$ \citep[e.g.][]{2011arXiv1103.0948C,falconi:012303}.

\subsection{Basic notions on symmetric Fock spaces.}
\label{sec:main-results-this}

We would like to extend the convergence results above to other types of states (that include $\varphi^{\otimes_n}$ as a particular case). We will use the Fock space method, so we recall some basic concepts of Fock spaces.

Let $\mathscr{H}$ be a separable Hilbert space, with scalar product
\begin{equation*}
  \braket{f}{g}_{\mathscr{H}}=\int\ide{x}\bar{f}(x)g(x)\; .
\end{equation*}
From $\mathscr{H}$, construct the symmetric Fock space $\mathscr{F}_s(\mathscr{H})$ as follows. Define:
\begin{gather*}
   \mathscr{H}_0=\mathds{C}\; ,\;\mathscr{H}_n:=S_n\Bigl(\mathscr{H}\otimes\dotsm\otimes\mathscr{H}\Bigr)\\
\intertext{with $S_n$ orthogonal symmetrizer on the $n$-th fold tensor product of $\mathscr{H}$, i.e.}
 \mathscr{H}_n=\Bigl\{\phi_n(x_1,\dotsc,x_n)\Bigl |\text{ $\phi_n$ is invariant for any permutation of variables}\Bigr\}\; .\\
\intertext{Then}
\mathscr{F}_s(\mathscr{H}):=\bigoplus_{n=0}^\infty \mathscr{H}_n\; .
\end{gather*}
Let $\phi=(\phi_0,\phi_1,\dotsc,\phi_n,\dotsc)$ be a vector of $\mathscr{F}_s(\mathscr{H})$. Then $\mathscr{F}_s(\mathscr{H})$, endowed with the norm
\begin{equation*}
  \norm{\phi}:=\Bigl(\sum_{n=0}^\infty\norm{\phi_n}^2_{\mathscr{H}_n} \Bigr)^{1/2}
\end{equation*}
is a (separable) Hilbert space. The vector $\Omega=(1,0,0,\dotsc)$ plays a special role in Fock spaces and it is often called the vacuum.

The basic operators of Fock spaces are the annihilation and creation operators. They are the adjoint of one another and are defined as following. Let $f\in\mathscr{H}$, $\phi=(\phi_0,\phi_1,\dotsc,\phi_n,\dotsc)$; then
  \begin{align*}
    a(f)&=\int\ide{x}f(x)a(x)\\
    a^*(f)&=\int\ide{x}f(x)a^*(x)\; ;
  \end{align*}
with
\begin{align*}
  (a(x)\phi)_n(x_1,\dotsc,x_n)&=\sqrt{n+1}\phi_{n+1}(x,x_1,\dotsc,x_n)\\
  (a^*(x)\phi)_n(x_1,\dotsc,x_n)&=\frac{1}{\sqrt{n}}\sum_{i=1}^n\delta(x-x_j)\phi_{n-1}(x_1,\dotsc,\hat{x}_j,\dotsc,x_n)\; ,
\end{align*}
where $\hat{x}_i$ means such variable is omitted. The $a^\#$ satisfy the following commutation properties:
\begin{equation*}
  [a(x),a^*(x')]=\delta(x-x')\; ,\; [a(x),a(x')]=[a^*(x),a^*(x')]=0\; .
\end{equation*}

From annihilation and creation operators we can construct an unitary operator, called the Weyl operator that plays a crucial role in our treatment of mean field limits. Let $\alpha\in\mathscr{H}$; then the Weyl operator, denoted by $C(\alpha)$, is defined as follows:
  \begin{equation*}
    C(\alpha)=\exp\bigl(a^*(\alpha)-a(\bar{\alpha})\bigr)\; .
  \end{equation*}
For all $\alpha\in\mathscr{H}$, $C(\alpha)$ is unitary on $\mathscr{F}_s(\mathscr{H})$. In addition to its unitarity, we will use the following properties:
\begin{enumerate}[i.]
\item\label{item:1} $C^*(\alpha)a(x)C(\alpha)=a(x)+\alpha(x)$, and therefore $C^*(\alpha)a^*(x)C(\alpha)=a^*(x)+\bar{\alpha}(x)$;
\item\label{item:2} $C(\alpha)C(\beta)=C(\alpha+\beta)e^{-i\Im \braket{\alpha}{\beta}}$;
\item\label{item:3} $C(\alpha)=e^{-\lVert\alpha\rVert^2/2}\exp\{a^*(\alpha)\}\exp\{-a(\alpha)\}$.
\end{enumerate}

Another useful class of operators are the ones defined by the so-called second quantization \citep[see e.g.][Chapter X.7]{MR0493420}. Let $A$ be a self-adjoint operator on $\mathscr{H}$, with domain of essential self-adjointness $D$. We define the second quantization of $A$, denoted by $\de\Gamma(A)$, as the following operator of $\mathscr{F}_s(\mathscr{H})$: let $\phi=(\phi_0,\phi_1,\dotsc,\phi_n,\dotsc)\in\mathscr{F}_s(\mathscr{H})$, then
  \begin{equation*}
    (\de\Gamma(A)\phi)_n(x_1,\dotsc,x_n)=\sum_{i=1}^n A(x_i)\phi_n(x_1,\dotsc,x_n)\; ,
  \end{equation*}
where $A(x_i)$ denotes the operator $A$ acting on the subspace of $\mathscr{H}_n$ corresponding to the $i$-th variable. $\de\Gamma(A)$ is essentially self-adjoint on the domain
\begin{equation*}
  D_A=\Bigl\{ \phi\in\mathscr{F}_s(\mathscr{H})\Bigl | \exists \tilde{n}, \forall n\geq\tilde{n} \; \phi_n=0\; ;\; \forall m\in\mathds{N} \; \phi_m\in D\otimes\dotsm \otimes D \Bigr\}\; .
\end{equation*}

The last notion we introduce is that of number operator $N$. For all $n\in\mathds{N}$, every $\phi_n\in\mathscr{H}_n$ is an eigenvector of $N$ with eigenvalue $n$. Precisely, let $\phi=(\phi_0,\phi_1,\dotsc,\phi_n,\dotsc)\in\mathscr{F}_s(\mathscr{H})$, then $N$ is defined as
  \begin{equation*}
    (N\phi_n)_n(x_1,\dotsc,x_n)=n\phi_n(x_1,\dotsc,x_n)\; .
  \end{equation*}
$N$ is a self-adjoint operator with domain
\begin{equation*}
  D(N)=\Bigl\{\phi\in\mathscr{F}_s(\mathscr{H}),\sum_{n=0}^\infty n^2\norm{\phi_n}^2<\infty\Bigr\}\; ;
\end{equation*}
it satisfies the following properties:
\begin{enumerate}[i.]
\item $N=\de\Gamma(1)=\int\ide{x}a^*(x)a(x)$.
\item Let $f\in\mathscr{H}$; then for all $\phi\in D(N^{1/2})$:
  \begin{equation*}
    \norm{a^\#(f)\phi}\leq \norm{f}_{\mathscr{H}}\norm{(N+1)^{1/2}\phi}\; .
  \end{equation*}
Furthermore, let $E_N(\lambda)$ be the spectral family of the operator $N$. Then for all $E$-measurable operator valued function $g$, we have
\begin{gather*}
  g(N)a(x)=a(x)g(N-1)\\
g(N)a^*(x)=a^*(x)g(N+1)\; ;
\end{gather*}
on suitable domains.
\item\label{item:7} Let $f\in\mathscr{H}\otimes\mathscr{H}$; then for all $\phi\in D(N)$:
  \begin{align*}
    \norm{\int\de{x}\ide{y}f(x,y)a^\#(x)a^\#(y)\phi}&\leq \norm{f}_{\mathscr{H}\otimes\mathscr{H}}\norm{(N+1)\phi}\\
    \norm{\int\de{x}\ide{y}f(x,y)a^*(x)a(y)\phi}&\leq \norm{f}_{\mathscr{H}\otimes\mathscr{H}}\norm{N\phi}\; .
  \end{align*}
\end{enumerate}

\subsection{Main results.}
\label{sec:quantum-system}

The quantum system we would like to study describes $n$ non-relativistic interacting bosons; its dynamics is dictated by the following Hamiltonian of $L^2_s(\mathds{R}^{3n})$: let $V$ be a real and symmetric function of $\mathds{R}^3$ (other assumptions on the potential will be specified later), then
\begin{equation}\label{eq:10}
  H_n=H_{0n}+V_n=\sum_{j=1}^n-\Delta_{x_j}+\frac{1}{n}\sum_{i<j}^n V(x_i-x_j)\; .
\end{equation}
This operator can be written in the language of second quantization on $\mathscr{F}_s(L^2(\mathds{R}^3))$ as:
\begin{equation}\label{eq:2}
  H=\int\ide{x}(\nabla a)^*(x)\nabla a(x)+\frac{1}{2n}\int\de{x}\ide{y} V(x-y)a^*(x)a^*(y)a(x)a(y)\; ;
\end{equation}
$H_n$ and $H$ agree on $\mathscr{H}_n$. Let now $\phi\in\mathscr{H}_n$ with $\norm{\phi}_{\mathscr{H}_n}=1$; we denote by $\phi(t)$ the time evolution of $\phi$ and by
\begin{equation*}
  \rho_\phi(t)=\ketbra{\phi(t)}{\phi(t)}
\end{equation*}
the corresponding density matrix, with integral kernel
\begin{equation*}
 \rho_\phi(t,x_1,\dotsc,x_n;y_1,\dotsc,y_n)=\phi(t,x_1,\dotsc,x_n)\bar{\phi}(t,y_1,\dotsc,y_n)\; . 
\end{equation*}
Also, we denote by $\Tr_1\rho_{\phi}(t)$ the one-particle reduced density matrix, with integral kernel
\begin{equation*}
  \Tr_1\rho_{\phi}(t,x;y)=\int\de{x_1}\dotsm\ide{x_{n-1}}\phi(t,x,x_1,\dotsc,x_{n-1})\bar{\phi}(t,y,x_1,\dotsc,x_{n-1})\; .
\end{equation*}

In the mean field limit $n\to\infty$, using suitable initial states, the one-particle reduced density matrix is expected to converge in some sense to the solution of Hartree equation:
\begin{equation}\label{eq:4}
  i\partial_t \varphi_t = -\Delta\varphi_t +(V*\ass{\varphi_t}^2)\varphi_t\; .
\end{equation}
As discussed above, such convergence has already been proved for factorized ($\varphi^{\otimes_n}$) and coherent ($C(\sqrt{n}\varphi)\Omega$) states. In this paper we prove that mean field limit convergence can be obtained for a wider class of states. First of all we consider states with $n$ particles, only a fraction of which is factorized in the same state $\varphi$ \citep[see also][where these states are used to construct an isomorphism between $\mathscr{H}_n$ and the truncated Fock space orthogonal to $\varphi$]{2012arXiv1211.2778L}. The precise definition is the following:
\begin{definition}[$\theta_{n,m}$]
  \label{sec:properties-theta_n-m-1}
Let $\mathscr{H}$ be a Hilbert space; $\mathscr{F}_s(\mathscr{H})=\bigoplus_{n=0}^\infty\mathscr{H}_n$ the corresponding symmetric Fock space. We also denote by $\mathscr{D}\subseteq\mathscr{H}$ a subspace of $\mathscr{H}$.

Now let $\varphi\in \mathscr{D}$ such that $\braket{\varphi}{\varphi}_{\mathscr{H}}=1$ and $\psi_m\in\mathscr{H}_m$ such that $\norm{\psi_m}_{\mathscr{H}_m}=1$ and
\begin{equation}
  \label{eq:1}
  \braket{\varphi}{\psi_m}_{\mathscr{H}_1}=\int\ide{x}\bar{\varphi}(x)\psi_m(x,x_1,\dotsc,x_{m-1})=0\; .
\end{equation}
We define
\begin{equation*}
  \theta_{n,m}:= c_{n,m}S_{n}(\varphi^{\otimes_{n-m}}\otimes\psi_m)\in\mathscr{H}_n\; ;
\end{equation*}
where:
\begin{gather*}
  \varphi^{\otimes_j}=\underset{j}{\underbrace{\varphi\otimes\dotsm\otimes\varphi}}\; ,\\
S_j:\underset{j}{\underbrace{\mathscr{H}\otimes\dotsm\otimes\mathscr{H}}}\longrightarrow \mathscr{H}_j\text{ orthogonal projector (symmetrizer),}\\
c_{n,m}=\binom{n}{m}^{1/2}\text{ (such that $\norm{\theta_{n,m}}=1$).}
\end{gather*}
\end{definition}
The convergence result about $\theta$-vectors is formulated in the following theorem:
\begin{thm}\label{sec:main-results-1}
Suppose there exists $D>0$ such that the operator inequality
\begin{equation*}
 V^2(x)\leq D(1-\Delta_x) 
\end{equation*}
is satisfied. Let $\theta_{n,m}$ satisfy definition~\ref{sec:properties-theta_n-m-1} with $\varphi\in H^1(\mathds{R}^3)$. Also, let $\varphi_t$ be the $\mathscr{C}^0(\mathds{R},H^1(\mathds{R}^3))$ solution of Hartree equation with initial datum $\varphi(0)=\varphi$. Then $\forall t \in\mathds{R}$:
  \begin{equation*}
\Tr \ass{\Tr_1\rho_{\theta_{n,m}}(t)-\ketbra{\varphi_t}{\varphi_t}}\leq 2 K_1 e^{K_2\ass{t}}\frac{1}{\sqrt{n}}e^{m/2}(m+1)^7\; ;
  \end{equation*}
where the $K_i$, $i=1,2$, are positive and depend only on $D$ and $\norm{\varphi}_{H^1}$.
\end{thm}
\begin{remark*}
  For technical reasons (see lemma~\ref{sec:properties-theta_n-m-3} below), the result above holds only if $m\leq \sqrt{7+3n}-3$. However since we are considering the limit $n\to\infty$, a stricter bound for $m$ has to be imposed if we want convergence to hold. In particular,
  \begin{equation*}
    \underset{n\to\infty}{\text{Tr$-$lim}}\; \Tr_1\rho_{\theta_{n,m}}(t) = \ketbra{\varphi_t}{\varphi_t}
  \end{equation*}
whenever exists $0\leq a<1$ such that, for large $n$,
\begin{equation*}
  m\sim a\ln n\; .
\end{equation*}
\end{remark*}

Also, we study the mean field limit for superpositions of $\varphi^{\otimes_n}$, $\theta_{n,m}$ or $C(\sqrt{n}\varphi)\Omega$ states. Such linear combinations have to satisfy the following definition:
\begin{definition}[$\Phi$, $\Theta$, $\Psi$]\label{sec:main-results-2}
Let $\mathscr{H}$ be a Hilbert space; $\mathscr{F}_s(\mathscr{H})=\bigoplus_{n=0}^\infty\mathscr{H}_n$ the corresponding symmetric Fock space. We also denote by $\mathscr{D}\subseteq\mathscr{H}$ a subspace of $\mathscr{H}$. Furthermore, assume:
\begin{enumerate}[i.]
\item\label{item:4} $(\alpha_i)_{i\in \mathds{N}},(\beta_i)_{i\in \mathds{N}},(\gamma_i)_{i\in \mathds{N}}\in l^1$.
\item\label{item:5} $(\varphi^{(i)})_{i\in\mathds{N}}$ such that $\varphi^{(i)}\in \mathscr{D}$, $\forall i\in\mathds{N}$ and $\sup_{i\in\mathds{N}}\norm{\varphi^{(i)}}_{\mathscr{D}}= M<+\infty$. Also we ask either
  \begin{enumerate}[a)]
  \item\label{item:8} $\norm{\varphi^{(i)}}_{\mathscr{H}}=1$, $\varphi^{(i)}$ and $\varphi^{(j)}$ linearly independent for all $i\neq j$; or
  \item\label{item:9} $\varphi^{(i)}\neq\varphi^{(j)}$ for all $i\neq j$. 
  \end{enumerate}
\item\label{item:6} To each vector of $(\varphi^{(i)})_{i\in\mathds{N}}$ satisfying \ref{item:8} we associate $(\varphi^{(i)})^{\otimes_n}\in\mathscr{H}_n$ and $\theta^{(i)}_{n,m_i}$ satisfying definition~\ref{sec:properties-theta_n-m-1}, such that $m_i\leq m_j$ if $i\leq j$, $m_i\leq m$ for all $i\in\mathds{N}$. To each vector of $(\varphi^{(i)})_{i\in\mathds{N}}$ satisfying \ref{item:9} we associate the Weyl operator $C(\sqrt{n}\varphi^{(i)})$.
\end{enumerate}
Then define
\begin{equation*}
  \Phi=\sum_{i\in\mathds{N}}\alpha_i(n)\, (\varphi^{(i)})^{\otimes_n}\; ,\; \Theta=\sum_{i\in\mathds{N}}\beta_i(n)\, \theta_{n,m_i}^{(i)}\; \Psi=\sum_{i\in\mathds{N}}\gamma_i(n) C(\sqrt{n}\varphi^{(i)})\Omega\; .
\end{equation*}
The suites $(\alpha_i(n))_{i\in\mathds{N}}$, $(\beta_i(n))_{i\in\mathds{N}}$ and $(\gamma_i(n))_{i\in\mathds{N}}$ are chosen such that $\norm{\Phi}=\norm{\Theta}=\norm{\Psi}=1$.
\end{definition}
\begin{remark}\label{sec:main-results-3}
  In particular we have
\begin{align*}
  \alpha_i(n)&=\alpha_i\Bigl(\sum_{i,j\in\mathds{N}}\bar{\alpha}_i\alpha_j\braket{\varphi^{(i)}}{\varphi^{(j)}}^n\Bigr)^{-1/2}\; ,\\
  \beta_i(n)&=\beta_i\Bigl(\sum_{i,j\in\mathds{N}}\bar{\beta}_i\beta_j\braket{\theta_{n,m_i}^{(i)}}{\theta_{n,m_j}^{(j)}}\Bigr)^{-1/2}\; ,\\
  \gamma_i(n)&=\gamma_i\Bigl(\sum_{i,j\in\mathds{N}}\bar{\gamma}_i\gamma_j\braket{C(\sqrt{n}\varphi^{(i)})\Omega}{C(\sqrt{n}\varphi^{(j)})\Omega}\Bigr)^{-1/2}\; .
\end{align*}
For all $i\in\mathds{N}$, $\alpha_i(n)$, $\beta_i(n)$ and $\gamma_i(n)$ are convergent when $n\to\infty$, and their absolute value is uniformly bounded in $n$ respectively by $K_\alpha\ass{\alpha_i}$, $K_\beta\ass{\beta_i}$ and $K_\gamma\ass{\gamma_i}$, where the constants depend only on the $l^1$-norm of the respective suite. Further details are discussed in section \ref{sec:proof-theorem-1}.
\end{remark}
\begin{remark}
  Since $\Psi$ states do not belong to a fixed particle subspace, we define the integral kernel of the reduced density matrix to be
\begin{equation*}
  \Tr_1 \rho_{\Psi}(t,x;y)=\frac{1}{\braket{\Psi(t)}{N\Psi(t)}}\braket{\Psi(t)}{a^*(y)a(x)\Psi(t)}\; .
\end{equation*}
\end{remark}

With a linear superposition of states as initial condition, the mean field limit is not a pure state. The reduced density matrix converges to a linear combination of projections on the solution of Hartree equation corresponding to each initial datum $\varphi^{(i)}$; in the topology induced by the Hilbert-Schmidt norm ($\norm{\,\cdot\,}_{HS}$):
\begin{thm}\label{sec:main-results}
Suppose there exists $D>0$ such that the operator inequality
\begin{equation*}
 V^2(x)\leq D(1-\Delta_x) 
\end{equation*}
is satisfied. Let $\Phi$, $\Theta$ and $\Psi$ satisfy definition~\ref{sec:main-results-2} with $\mathscr{D}=H^1(\mathds{R}^3)$. Also, for all $i\in\mathds{N}$, let $\varphi^{(i)}_t$ be the $\mathscr{C}^0(\mathds{R},H^1(\mathds{R}^3))$ solution of Hartree equation with initial datum $\varphi^{(i)}(0)=\varphi^{(i)}$. Then $\forall t \in\mathds{R}$, if there is $0\leq a<1/2$ such that, for large $n$, $m\sim a\ln n$:
\begin{align*}
  \underset{n\to\infty}{\text{HS-}\mathrm{lim}}\; &\Tr_1\rho_{\Phi}(t)=\sum_{i\in\mathds{N}}\ass{\frac{\alpha_i}{\norm{(\alpha_i)_{i\in\mathds{N}}}_{l^2}}}^2\ketbra{\varphi_t^{(i)}}{\varphi_t^{(i)}}\; ;\\
  \underset{n\to\infty}{\text{HS-}\mathrm{lim}}\; &\Tr_1\rho_{\Theta}(t)=\sum_{i\in\mathds{N}}\ass{\frac{\beta_i}{\norm{(\beta_i)_{i\in\mathds{N}}}_{l^2}}}^2\ketbra{\varphi_t^{(i)}}{\varphi_t^{(i)}}\; ;\\
  \underset{n\to\infty}{\text{HS-}\mathrm{lim}}\; &\Tr_1\rho_{\Psi}(t)=\sum_{i\in\mathds{N}}\ass{\frac{\gamma_i}{\norm{(\gamma_i)_{i\in\mathds{N}}}_{l^2}}}^2\ketbra{\varphi_t^{(i)}}{\varphi_t^{(i)}}\; .
\end{align*}
In particular, the following bounds hold:
  \begin{gather}\label{eq:5}\begin{split}
\norm{\Tr_1\rho_{\Phi}(t)-\norm{(\alpha_i)_{i\in\mathds{N}}}^{-2}_{l^2}\sum_{i\in\mathds{N}}\ass{\alpha_i}^2\ketbra{\varphi_t^{(i)}}{\varphi_t^{(i)}}}_{HS}\leq K_1\sum_{i<j}\ass{\bar{\alpha}_i(n)\alpha_j(n)}\ass{\braket{\varphi^{(i)}}{\varphi^{(j)}}}^n \\+ K_2 e^{K_3\ass{t}}\frac{1}{n^{1/4}}\; ,
  \end{split}\\
\label{eq:11}\begin{split}
\norm{\Tr_1\rho_{\Theta}(t)-\norm{(\beta_i)_{i\in\mathds{N}}}^{-2}_{l^2}\sum_{i\in\mathds{N}}\ass{\beta_i}^2\ketbra{\varphi_t^{(i)}}{\varphi_t^{(i)}}}_{HS}\leq K_1\sum_{i<j}\ass{\bar{\beta}_i(n)\beta_j(n)}\ass{\braket{\theta_{n,m_i}^{(i)}}{\theta_{n,m_j}^{(j)}}} \\+ K_2 e^{K_3\ass{t}}\frac{1}{n^{1/4}}e^{m/2}(m+1)^3\; ;
  \end{split}\\
\label{eq:12}
\begin{split}
  \norm{\Tr_1\rho_{\Psi}(t)-\norm{(\gamma_i)_{i\in\mathds{N}}}^{-2}_{l^2}\sum_{i\in\mathds{N}}\ass{\gamma_i}^2\ketbra{\varphi_t^{(i)}}{\varphi_t^{(i)}}}_{HS}\leq K_1e^{-4M^2 n} + K_2 e^{K_3\ass{t}}\frac{1}{n^{1/2}}\; ;
\end{split}
\end{gather}
where the $K_i$, $i=1,2,3$, are positive and depend on $D$, $M$ and $\norm{(\,\cdot\,_i)_{i\in\mathds{N}}}_{l^1}$.
\end{thm}
\begin{remark}
The first term on the right hand side of equation~\eqref{eq:5} converges to zero when $n\to\infty$ because, by definition~\ref{sec:main-results-2} and Riesz's lemma, $\ass{\braket{\varphi^{(i)}}{\varphi^{(j)}}}<1$ (and $\ass{\bar{\alpha}_i(n)\alpha_j(n)}\leq K_\alpha^2 \ass{\bar{\alpha}_i\alpha_j}$). The first term on the right hand side of equation~\eqref{eq:11} also converges to zero if $m\leq \ln n$. This is because
\begin{equation*}
  \ass{\braket{\theta_{n,m_i}^{(i)}}{\theta_{n,m_j}^{(j)}}}\leq (m+1)(m!)^2 n^m\ass{\braket{\varphi^{(i)}}{\varphi^{(j)}}}^{n-2m}\leq K n^{2(\ln n+1)}\ass{\braket{\varphi^{(i)}}{\varphi^{(j)}}}^{n-2\ln n}\underset{n\to\infty}{\longrightarrow}0\; ;
\end{equation*}
furthermore since $\ass{\bar{\beta}_i(n)\beta_j(n)}\ass{\braket{\theta_{n,m_i}^{(i)}}{\theta_{n,m_j}^{(j)}}}\leq K_\beta^2\ass{\bar{\beta}_i\beta_j}$ we can exchange summation with the limit $n\to\infty$.
\end{remark}
\begin{remark}
  In this paper we focused attention only on one-particle reduced density matrices. The same method can be used to calculate the limit of $k$-particle reduced density matrices (with $k>1$).
\end{remark}

The rest of the paper is organized as follows. In section~\ref{sec:dynam-prop} we analyze the dynamics of classical (Hartree) and quantum system. In section~\ref{sec:properties-theta_n-m} the combinatorial properties of $\theta_{n,m}$ states are studied. Finally in section~\ref{sec:mean-field-limit} we consider the limit $n\to\infty$, and prove theorems~\ref{sec:main-results-1} and~\ref{sec:main-results}.

\section{Classical and quantum dynamics.}
\label{sec:dynam-prop}

In this section we review some properties of quantum and classical evolution needed to perform the classical limit.

The first proposition concerns the existence and unicity of solutions of the Hartree Cauchy problem \citep[see][Remark 1.3]{MR2530155}:
\begin{equation}\label{eq:3}
    \left\{
    \begin{aligned}
        i\partial_t \varphi_t &= -\Delta\varphi_t +(V*\ass{\varphi_t}^2)\varphi_t\\
      \varphi(t_0)&=\varphi_0
    \end{aligned}
\right .\mspace{60mu} .
\end{equation}
\begin{proposition}\label{sec:class-quant-dynam}
  Under the assumption $V^2(x)\leq D(1-\Delta_x)$, \eqref{eq:3} has a unique solution $\varphi_t\in \mathscr{C}^0(\mathds{R},H^1(\mathds{R}^3))$ for all $\varphi_0\in H^1(\mathds{R}^3)$. Furthermore, using also the conservation in time of $\norm{\varphi}_{L^2}$ and energy
  \begin{equation*}
    E(\varphi)=\int\ide{x}\ass{\nabla\varphi(x)}^2+\frac{1}{2}\int\de{x}\ide{y} V(x-y)\norm{\varphi(x)}^2\norm{\varphi(y)}^2\; ,
  \end{equation*}
we obtain the following bound for $\norm{\varphi_t}_{H^1}$:
\begin{equation*}
  \norm{\varphi_t}_{H^1}^2\leq c \, \Bigl(\norm{\varphi_0}^2_{H^1}(1+\norm{\varphi_0}^2_{L^2})+\norm{\varphi_0}_{L^2}^4+\norm{\varphi_0}_{L^2}^2\Bigr)\; ,
\end{equation*}
for some positive $c$ that depends only on $D$.
\end{proposition}

We can formulate also the following proposition concerning quantum evolution \citep[see][]{MR530915,MR2530155}.
\begin{proposition}\label{sec:class-quant-dynam-1}
  Let $V$ be a real symmetric function such that for some $0<D$ the operator inequality $V^2(x)\leq D(1-\Delta_x)$ is satisfied. Then:
  \begin{enumerate}[i.]
  \item $H$ is a self-adjoint operator with domain $D(H)\subset\mathscr{F}_s(L^2(\mathds{R}^3))$, defined through the direct sum decomposition
    \begin{equation*}
      H=\bigoplus_{n=1}^\infty H_n\; .
    \end{equation*}
$H_n$ is the self-adjoint operator~\eqref{eq:10} on $D(H_{0n})$ (defined as a Kato sum).
  \item   Let $\varphi_0\in H^1(\mathds{R}^3)$, $\varphi_t\in \mathscr{C}^0(\mathds{R},H^1(\mathds{R}^3))$ the associated solution of~\eqref{eq:3}. Define also
  \begin{gather*}
    U(t)=\exp(-it H)\; ,\\
    W(t,t_0)=C^*(\sqrt{n}\varphi_t)U(t-t_0)C(\sqrt{n}\varphi_0)\; .
  \end{gather*}
Then for all $\delta\geq 0$, $\phi\in D(N^{\max\{2,2\delta +1\}})$ the following inequality holds:
\begin{equation}\label{eq:9}
      \norm{(N+1)^\delta W(t,t_0)\phi}\leq K_1(\delta)e^{K_2(\delta)\ass{t-t_0}}\norm{(N+1)^{\max\{2,2\delta+1\}}\phi}\; ;
\end{equation}
where $K_i(\delta)$, $i=1,2$, are positive and depend only on $\delta$ and $\norm{\varphi_0}_{H^1}$.
  \end{enumerate}
\end{proposition}

\section{Properties of $\theta_{n,m}$ states.}
\label{sec:properties-theta_n-m}

In this section we study the partially factorized states $\theta_{n,m}$ defined in section~\ref{sec:introduction}. The technical result formulated in lemma~\ref{sec:properties-theta_n-m-3} is crucial to perform the mean field limit $n\to\infty$. Its proof mimics the one performed in the case of completely condensed states, found in \citep[Lemma 6.3]{2011JMP}.

\begin{lemma}
  \label{sec:properties-theta_n-m-2}
The following identities hold:
\begin{align*}
  \theta_{n,m}&=\frac{1}{\sqrt{m!}}\int\de{y_1}\dotsm\ide{y_m}\psi_m(y_1,\dotsc,y_m)a^*(y_1)\dotsm a^*(y_m)\varphi^{\otimes_{n-m}}\\
&=\frac{1}{\sqrt{(n-m)!m!}}\int\de{y_1}\dotsm\ide{y_m}\psi_m(y_1,\dotsc,y_m)\bigl(a^*(\varphi)\bigr)^{n-m}a^*(y_1)\dotsm a^*(y_m)\Omega\\
&=\frac{d_{n,m}}{\sqrt{m!}}\int\de{y_1}\dotsm\ide{y_m}\psi_m(y_1,\dotsc,y_m)P_{n}C(\sqrt{n}\varphi)a^*(y_1)\dotsm a^*(y_m)\Omega\; ;
\end{align*}
where $P_{n}$ is the projector on $\mathscr{H}_{n}$ and
\begin{equation*}
  d_{n,m}=\frac{\sqrt{(n-m)!}}{\exp(-n/2)n^{(n-m)/2}} \text{ ($d_{n,m}\sim (n-m)^{1/4}e^{m/2}$ for large $n$).}
\end{equation*}
\end{lemma}
\begin{proof}
First two equalities are easy to prove. To prove the last one, recall the formula 
\begin{equation*}
  \varphi^{\otimes_{n-m}}=d_{n,m} P_{n-m} C(\sqrt{n}\varphi)\Omega\; .
\end{equation*}
We then obtain:
\begin{equation*}
  \begin{split}
    \theta_{n,m}=\frac{d_{n,m}}{\sqrt{m!}}\int\de{y_1}\dotsm\ide{y_m}\psi_m(y_1,\dotsc,y_m) P_{n} a^*(y_1)\dotsm a^*(y_m)P_{n-m} C(\sqrt{n}\varphi)\Omega\\
=\frac{d_{n,m}}{\sqrt{m!}}\int\de{y_1}\dotsm\ide{y_m}\psi_m(y_1,\dotsc,y_m) P_{n} C(\sqrt{n}\varphi)(a^*(y_1)+\sqrt{n}\bar{\varphi}(y_1))\dotsm \\
(a^*(y_m)+\sqrt{n}\bar{\varphi}(y_m))\Omega\; ;
  \end{split}
\end{equation*}
now all the terms containing at least a $\bar{\varphi}$ vanish by~\eqref{eq:1}.
\end{proof}

\begin{lemma}
  \label{sec:properties-theta_n-m-3}
Let $\theta_{n,m}$ satisfy definition~\ref{sec:properties-theta_n-m-1}, with $m\leq \sqrt{7+3n}-3$. Then $\forall\epsilon >0$:
\begin{equation*}
  \norm{(N+1)^{-(1/4 +\epsilon)}C^*(\sqrt{n}\varphi)\theta_{n,m}}\leq \frac{K(\epsilon)}{d_{n,m}}e^{m/2}\; ;
\end{equation*}
for some positive $K(\epsilon)$ that depends only on $\epsilon$.
\end{lemma}
\begin{proof}
  Consider $C^*(\sqrt{n}\varphi)\theta_{n,m}$; by definition of Weyl operators, we obtain $\forall j\leq n+m$:
  \begin{equation*}
    \begin{split}
      \Bigl( C^*(\sqrt{n}\varphi)\theta_{n,m}\Bigr)_j=e^{-n/2}\sum_{i=0}^j\frac{(-\sqrt{n})^{i}}{i!}\bigl(a^*(\varphi) \bigr)^i\frac{(\sqrt{n})^{n-j+i}}{(n-j+i)!}\bigl(a(\bar{\varphi}) \bigr)^{n-j+i}\theta_{n,m}\\
=e^{-n/2}(\sqrt{n})^{n-j}\sum_{i=0}^j(-1)^i\frac{n^{i}}{i!(n-j+i)!}\bigl(a^*(\varphi) \bigr)^i\bigl(a(\bar{\varphi}) \bigr)^{n-j+i}\theta_{n,m}\; .
    \end{split}
  \end{equation*}
By lemma~\ref{sec:properties-theta_n-m-2} (second equality) we write
\begin{equation*}
  \begin{split}
    \bigl(a^*(\varphi) \bigr)^i\bigl(a(\bar{\varphi}) \bigr)^{n-j+i}\theta_{n,m}=\frac{1}{\sqrt{(n-m)!m!}}\int\de{y_1}\dotsm\ide{y_m}\psi_m(y_1,\dotsc,y_m) \bigl(a^*(\varphi) \bigr)^i\bigl(a(\bar{\varphi}) \bigr)^{n-j+i}\\
\bigl(a^*(\varphi)\bigr)^{n-m}a^*(y_1)\dotsm a^*(y_m)\Omega\; ;
  \end{split}
\end{equation*}
so if $n-j+i>n-m$, i.e. $m-j+i>0$, then either $a(\bar{\varphi})$ acts on some $a^*(y)$, or on $\Omega$, yielding zero (by~\eqref{eq:1}, and definition of $\Omega$). Therefore we must have $i\leq j-m$, and consequently $j\geq m$. In such case we obtain
\begin{equation*}
  \bigl(a^*(\varphi) \bigr)^i\bigl(a(\bar{\varphi}) \bigr)^{n-j+i}\theta_{n,m}=\frac{\sqrt{(n-m)!(j-m)!}}{(j-m-i)!}\theta_{j-m,m}\; ;
\end{equation*}
this leads to
\begin{equation*}
  \Bigl( C^*(\sqrt{n}\varphi)\theta_{n,m}\Bigr)_j=\left\{
  \begin{aligned}
    &0\text{ if $j<m$}\\
    &\sum_{i=0}^{j-m}(-1)^in^{i}\frac{e^{-n/2}(\sqrt{n})^{n-j}\sqrt{(n-m)!(j-m)!}}{i!(j-m-i)!(n-j+i)!}\theta_{j-m,m}\text{ if $m\leq j\leq n$.}
  \end{aligned}\right .
\end{equation*}
So if we set $k=j-m$, for all $0\leq k\leq n-m$:
\begin{equation*}
  \Bigl( C^*(\sqrt{n}\varphi)\theta_{n,m}\Bigr)_{k+m}=e^{-n/2}(\sqrt{n})^{n-m-k}\sqrt{\frac{(n-m)!}{k!}}\biggl( \sum_{i=0}^k (-1)^i n^i\frac{k!}{i!(k-i)!(n-m-k+i)!}\biggr)\theta_{k,m}\; .
\end{equation*}
Now define
\begin{equation*}
  A_k:=\norm{\Bigl( C^*(\sqrt{n}\varphi)\theta_{n,m}\Bigr)_{k+m}}_{\mathscr{H}_{k+m}}\; ,
\end{equation*}
and the Laguerre polynomial
\begin{equation}\label{eq:6}
  L^{(\alpha)}_k(x):=\sum_{i=0}^k(-1)^ix^i\frac{(k+\alpha)!}{i!(k-i)!(\alpha+i)!} \; .
\end{equation}
Then we have
\begin{equation*}
  A_k=e^{-n/2}(\sqrt{n})^{n-m-k}\sqrt{\frac{k!}{(n-m)!}}\ass{L^{(n-m-k)}_k(n)}\text{ for all $0\leq k\leq n-m$.}
\end{equation*}
Actually, an explicit calculation shows that $A_0=1/d_{n,m}$, $A_1=m/\sqrt{n}d_{n,m}$. For $k\geq 2$, we use a sharp estimate for Laguerre polynomials obtained by \citet{MR2168916}: let $s=\sqrt{k+\alpha+1}+\sqrt{k}$, $q=\sqrt{k+\alpha+1}-\sqrt{k}$ and $r(x)=(x-q^2)(s^2-x)$; then for $\alpha>-1$, $k\geq 2$ and $x\in(q^2,s^2)$
\begin{equation}\label{eq:7}
  \ass{L^{(\alpha)}_k(x)}<\sqrt{\frac{(k+\alpha)!}{k!}}\sqrt{\frac{x(s^2-q^2)}{r(x)}}e^{x/2}x^{-(\alpha+1)/2}\; .
\end{equation}
In the case of $L^{(n-m-k)}_k(n)$ we obtain: $s=\sqrt{n-m+1}+\sqrt{k}$, $q=\sqrt{n-m+1}-\sqrt{k}$ and $r(n)=4k(n-m+1)-(k-m+1)^2$. The condition $n\in(q^2,s^2)$ implies that $m\leq \sqrt{7+3n}-3$. Therefore we obtain:
\begin{equation*}
  \begin{split}
      A_k<\sqrt{\frac{4\sqrt{k(n-m+1)}}{(2\sqrt{k(n-m+1)}+k-m+1)(2\sqrt{k(n-m+1)}-k+m-1)}}\; .
  \end{split}
\end{equation*}
Then, since $2\leq k\leq n-m$:
\begin{equation*}
  A_k<(\sqrt{2}/2-1/2)^{-1/2}(k+1)^{-1/4}(n-m+1)^{-1/4}\; .
\end{equation*}
Let $\delta>0$, and consider
\begin{equation*}
  \braket{C^*(\sqrt{n}\varphi)\theta_{n,m}}{(N+1)^{-\delta} C^*(\sqrt{n}\varphi)\theta_{n,m}}=\sum_{k=0}^\infty \frac{\ass{A_k}^2}{(k+m+1)^\delta}\; ,
\end{equation*}
since $\Bigl( C^*(\sqrt{n}\varphi)\theta_{n,m}\Bigr)_j=0$ for all $0\leq j< m$. Using the bound above, we obtain:
\begin{equation*}
  \begin{split}
    \braket{C^*(\sqrt{n}\varphi)\theta_{n,m}}{(N+1)^{-\delta} C^*(\sqrt{n}\varphi)\theta_{n,m}}\leq \biggl( \frac{1}{d_{n,m}^2}(1+m/\sqrt{n})+(\sqrt{2}/2-1/2)^{-1}\\(n-m+1)^{-1/2}\sum_{k=2}^{n-m}\frac{1}{(k+1)^{\delta+1/2}}\biggr)+\sum_{k=n-m+1}^\infty\frac{\ass{A_k}^2}{(k+1)^\delta}\; ;
  \end{split}
\end{equation*}
furthermore, since
\begin{equation*}
  \sum_{k=n-m+1}^\infty\ass{A_k}^2\leq \norm{C^*(\sqrt{n}\varphi)\theta_{n,m}}^2=1\; ,
\end{equation*}
we have
\begin{equation*}
  \begin{split}
    \braket{C^*(\sqrt{n}\varphi)\theta_{n,m}}{(N+1)^{-\delta} C^*(\sqrt{n}\varphi)\theta_{n,m}}\leq \biggl( \frac{1}{d_{n,m}^2}(1+m/\sqrt{n})+(\sqrt{2}/2-1/2)^{-1}\\(n-m+1)^{-1/2}\bigl(H_{\delta+1/2}(n-m+1)-1 \bigr)+(n-m+2)^{-\delta}\biggr)\; ,
  \end{split}
\end{equation*}
where $H_{\delta+1/2}(n-m+1)$ is a generalized harmonic number, whose limit for $n\to\infty$ converges only if $\delta+1/2>1$, i.e. $\delta>1/2$; and
\begin{equation*}
  \lim_{n\to\infty}H_{\delta+1/2}(n-m+1)=\zeta(\delta+1/2)\text{ ($\zeta$ is the Riemann zeta function).}
\end{equation*}
Then for all $\epsilon>0$ we can write
\begin{equation*}
  \begin{split}
    \braket{C^*(\sqrt{n}\varphi)\theta_{n,m}}{(N+1)^{-(1/4+\epsilon)} C^*(\sqrt{n}\varphi)\theta_{n,m}}\leq \biggl( \frac{1}{d_{n,m}^2}(1+m/\sqrt{n})+(\sqrt{2}/2-1/2)^{-1}\\\bigl(\zeta(1+\epsilon)-1 \bigr) (n-m+1)^{-1/2}+(n-m+2)^{-(1/2+2\epsilon)}\biggr)\; ;
  \end{split}
\end{equation*}
thus concluding the proof.
\end{proof}

\section{Mean field limit.}
\label{sec:mean-field-limit}

In this section we investigate the behaviour of the transition amplitude of a creation and an annihilation operator between time evolved $\theta_{n,m}$, $\Phi$ and $\Theta$ states, when $n$ is large. Hence we are able to prove theorems~\ref{sec:main-results-1} and~\ref{sec:main-results} on the convergence of one-particle reduced density matrices.

\subsection{Proof of Theorem~\ref{sec:main-results-1}.}
\label{sec:proof-theorem}

First of all we consider the transition amplitude $T_{\theta_{n,m}}(t)$. Let $\theta_{n,m}$ and $U(t)$ be as above; define
  \begin{equation*}
    T_{\theta_{n,m}}(t)\equiv T_{\theta_{n,m}}(t,x,y) := \braket{U(t)\theta_{n,m}}{a^*(x)a(y)U(t)\theta_{n,m}}\; .
  \end{equation*}
Then we can formulate the following proposition:
\begin{proposition}\label{sec:mean-field-limit-1}
Let $\varphi_t\in \mathscr{C}^0(\mathds{R},H^1(\mathds{R}^3))$ be the solution of Hartree equation with initial datum $\varphi\in H^1(\mathds{R}^3)$. Then $\forall t\in\mathds{R}$, $m\leq \sqrt{7+3n}-3$:
  \begin{equation*}
    \begin{split}
      \norm{T_{\theta_{n,m}}(t)-n\bar{\varphi}_t\varphi_t}_{L^2(\mathds{R}^6)}\leq K_1e^{K_2\ass{t}}(1+2\sqrt{n})e^{m/2}(m+1)^7\; ;
    \end{split}
  \end{equation*}
where the $K_i$, $i=1,2$, are positive and depend only on $\norm{\varphi}_{H^1}$.
\end{proposition}
\begin{proof}
    By lemma~\ref{sec:properties-theta_n-m-2} and property~\ref{item:1} of Weyl operators, we can write:
  \begin{equation*}
    \begin{split}
      T_{\theta_{n,m}}(t)=d_{n,m}\braket{U(t)\theta_{n,m}}{P_{n}C(\sqrt{n}\varphi_t)\bigl(a^*(x)+\sqrt{n}\bar{\varphi}_t(x)\bigr)\bigl(a(y)+\sqrt{n}\varphi_t(y)\bigr)C^*(\sqrt{n}\varphi_t)U(t)\\C(\sqrt{n}\varphi)\psi_m}\\
=n\bar{\varphi}_t(x)\varphi_t(y)+d_{n,m}\braket{(N+1)^{-(1/4+\epsilon)}C^*(\sqrt{n}\varphi)\theta_{n,m}}{(N+1)^{1/4+\epsilon}W^*(t,0)\\\bigl(a^*(x)a(y)+\sqrt{n}(a^*(x)\varphi_t(y)+a(y)\bar{\varphi}_t(x)) \bigr)W(t,0)\psi_m}\; .
    \end{split}
  \end{equation*}
Now, let $\xi\in L^2(\mathds{R}^6)$. Then by lemma~\ref{sec:properties-theta_n-m-3}, with $\epsilon\leq 1/4$:
\begin{equation*}
  \begin{split}
    \ass{\int\de{x}\ide{y}\bar{\xi}(x,y)\Bigl(T_{\theta_{n,m}}(t,x,y)-n\bar{\varphi}_t(x)\varphi_t(y) \Bigr)}\leq Ke^{m/2}\norm{(N+1)^{1/4+\epsilon}W^*(t,0)\int\de{x}\ide{y}\\\bar{\xi}(x,y)\bigl(a^*(x)a(y)+\sqrt{n}(a^*(x)\varphi_t(y)+a(y)\bar{\varphi}_t(x)) \bigr)W(t,0)\psi_m}\; .
\end{split}
\end{equation*}
Now using: \eqref{eq:9} two times and the properties of annihilation, creation and number operators listed in section~\ref{sec:main-results-this}, we obtain:
\begin{equation*}
  \begin{split}
    \ass{\int\de{x}\ide{y}\bar{\xi}(x,y)\Bigl(T_{\theta_{n,m}}(t,x,y)-n\bar{\varphi}_t(x)\varphi_t(y) \Bigr)}\leq K_1e^{K_2\ass{t}}(1+2\sqrt{n})e^{m/2}\norm{(N+1)^7\psi_m}\\\norm{\xi}_{L^2(\mathds{R}^6)}\; .
  \end{split}
\end{equation*}
The result is then proved applying Riesz's lemma on $L^2(\mathds{R}^6)$.
\end{proof}

The trace of $T_{\theta_{n,m}}(t)$ is defined as
  \begin{equation*}
  \Tr T_{\theta_{n,m}}(t)=\int\ide{x} T_{\theta_{n,m}}(t,x,x)=\braket{U(t)\theta_{n,m}}{N U(t)\theta_{n,m}}=n\; .
\end{equation*}
Consider now the one-particle reduced density matrix $\Tr_1\rho_{\theta_{n,m}}(t)$; its integral kernel can be written as
\begin{equation*}
  \Tr_1\rho_{\theta_{n,m}}(t,x;y)=\frac{T_{\theta_{n,m}}(t,y,x)}{\Tr T_{\theta_{n,m}}(t)}=\frac{T_{\theta_{n,m}}(t,y,x)}{n}\; .
\end{equation*}
Then, by proposition~\ref{sec:mean-field-limit-1} we immediately obtain
\begin{equation*}
  \norm{\Tr_1\rho_{\theta_{n,m}}(t)-\ketbra{\varphi_t}{\varphi_t}}_{HS}\leq K_1 e^{K_2\ass{t}}\Bigl(\frac{1}{n}+\frac{2}{\sqrt{n}}\Bigr)e^{m/2}(m+1)^7\; ,
\end{equation*}
where $\norm{\,\cdot\,}_{HS}$ stands for the Hilbert-Schmidt norm. Denote now $X=\Tr_1\rho_{\theta_{n,m}}(t)-\ketbra{\varphi_t}{\varphi_t}$. Since $\ketbra{\varphi_t}{\varphi_t}$ is a rank one projection, we have that $\Tr\ass{X} \leq 2\nnorm{X}$, where $\nnorm{\,\cdot\,}$ denotes the operator norm. Since $\nnorm{X}\leq \norm{X}_{HS}$, it follows that $\Tr\ass{X}\leq 2\norm{X}_{HS}$, hence the theorem is proved.

\subsection{Proof of Theorem~\ref{sec:main-results}.}
\label{sec:proof-theorem-1}

As discussed in remark~\ref{sec:main-results-3} of section~\ref{sec:quantum-system}, an $n$-dependent normalization has to be performed on linear superposition states. In the following lemma we study its behavior at large $n$.
\begin{lemma}\label{sec:proof-theor-refs}
Let $m$ be as in definition~\ref{sec:main-results-2}. Then for all $m\leq \ln n$, the following limits hold:
  \begin{gather*}
    \lim_{n\to\infty}\sum_{i,j\in\mathds{N}}\bar{\alpha}_i\alpha_j\braket{\varphi^{(i)}}{\varphi^{(j)}}^n=\sum_{i\in\mathds{N}}\ass{\alpha_i}^2\\
    \lim_{n\to\infty}\sum_{i,j\in\mathds{N}}\bar{\beta}_i\beta_j\braket{\theta_{n,m_i}^{(i)}}{\theta_{n,m_j}^{(j)}}=\sum_{i\in\mathds{N}}\ass{\beta_i}^2\\
    \begin{split}
      \lim_{n\to\infty}\sum_{i,j\in\mathds{N}}\bar{\gamma}_i\gamma_j\braket{C(\sqrt{n}\varphi^{(i)})\Omega}{C(\sqrt{n}\varphi^{(j)})\Omega}=\lim_{n\to\infty}\sum_{i,j\in\mathds{N}}\bar{\gamma}_i\gamma_j e^{in\Im\braket{\varphi^{(i)}}{\varphi^{(j)}}}e^{-n\norm{\varphi^{(j)}-\varphi^{(i)}}^2/2}\\=\sum_{i\in\mathds{N}}\ass{\gamma_i}^2\; .
    \end{split}
  \end{gather*}
Hence $\ass{\alpha_i(n)}$, $\ass{\beta_i(n)}$ and $\ass{\gamma_i(n)}$ are uniformly bounded in $n$.
\end{lemma}
\begin{proof}
We can take the limit $n\to\infty$ for each term of the summation by the dominated convergence theorem, because the absolute values of the $n$-dependent scalar products of vectors are bounded by $1$, and the suites $(\alpha_i)$, $(\beta_i)$ and $(\gamma_i)$ are in $l^1$. Then the first limit is obtained because $\braket{\varphi^{(i)}}{\varphi^{(i)}}=1$ for all $i\in\mathds{N}$, while, by definition~\ref{sec:main-results-2}, $\ass{\braket{\varphi^{(i)}}{\varphi^{(j)}}}<1$ for all $i\neq j$. The second limit is analogous, since $\braket{\theta_{n,m_i}^{(i)}}{\theta_{n,m_i}^{(i)}}=1$, while for all $m\leq \ln n$ we have $\ass{\braket{\theta_{n,m_i}^{(i)}}{\theta_{n,m_j}^{(j)}}}\leq K n^{2(\ln n+1)}\ass{\braket{\varphi^{(i)}}{\varphi^{(j)}}}^{n-2\ln n}\to 0$ when $n\to\infty$. The last one is obtained because, using the properties of Weyl operators, we have:
\begin{equation*}
  \braket{C(\sqrt{n}\varphi^{(i)})\Omega}{C(\sqrt{n}\varphi^{(j)})\Omega}=\braket{\Omega}{C(-\sqrt{n}\varphi^{(i)})C(\sqrt{n}\varphi^{(j)})\Omega}=e^{in\Im\braket{\varphi^{(i)}}{\varphi^{(j)}}}e^{-n\norm{\varphi^{(j)}-\varphi^{(i)}}^2/2}\; .
\end{equation*}
\end{proof}
The proof of theorem~\ref{sec:main-results} is similar to the one above for theorem~\ref{sec:main-results-1}; we carry it out explicitly only for $\Phi$ vectors, the other case being analogous. Define the transition amplitude
\begin{equation*}
  T_\Phi(t,x,y)=\braket{\Phi}{e^{itH}a^*(x)a(y)e^{-itH}\Phi}\; .
\end{equation*}
Let $\xi\in L^2(\mathds{R}^6)$. Then using property~\ref{item:7} of the number operator $N$ we obtain
\begin{equation*}
  \ass{\braket{\xi}{T_\Phi(t)}_{L^2(\mathds{R}^6)}}\leq (n+1)\norm{\xi}_{L^2(\mathds{R}^6)}\; ,
\end{equation*}
then $T_\Phi(t)\in L^2(\mathds{R}^6)$ for all $n\in\mathds{N}$. Also, since $\norm{\Phi}=1$ by definition~\ref{sec:main-results-2}, we have that
\begin{equation*}
  \Tr T_\Phi(t)=\int\ide{x}T_\Phi(t,x,x)=n\; .
\end{equation*}
Furthermore the series
\begin{equation*}
  \int\de{x}\ide{y}\bar{\xi}(x,y)\sum_{i\in\mathds{N}}\ass{\alpha_i(n)}^2\bar{\varphi}_t^{(i)}(x)\varphi_t^{(i)}(y)=\sum_{i\in\mathds{N}}\ass{\alpha_i(n)}^2\braket{\xi}{\bar{\varphi}_t^{(i)}\varphi_t^{(i)}}_{L^2(\mathds{R}^6)}
\end{equation*}
is absolutely convergent under the hypotheses of definition~\ref{sec:main-results-2}:
\begin{equation*}
  \sum_{i\in\mathds{N}}\ass{\alpha_i(n)}^2\ass{\braket{\xi}{\bar{\varphi}_t^{(i)}\varphi_t^{(i)}}_{L^2(\mathds{R}^6)}}\leq \norm{\xi}_{L^2(\mathds{R}^6)}\norm{(\alpha_i(n))_{i\in\mathds{N}}}^2_{l^2}\; .
\end{equation*}
Therefore we can write, with $d_n=\sqrt{n!}/\exp(-n/2)n^{n/2}$,
\begin{equation*}
  \begin{split}
    \ass{\int\de{x}\ide{y}\bar{\xi}(x,y)\Bigl(T_\Phi(t,x,y)-n\sum_{i\in\mathds{N}}\ass{\alpha_i(n)}^2\bar{\varphi}_t^{(i)}(x)\varphi_t^{(i)}(y) \Bigr)}\leq \sum_{i\in\mathds{N}}d_n\ass{\alpha_i(n)}^2\Biggl\lvert\braket{C^*(\sqrt{n}\varphi_0^{(i)})(\varphi^{(i)})^{\otimes_n}}{\\W^*(t,0)
\int\de{x}\ide{y}\bar{\xi}(x,y)\Bigl( \sqrt{n}(\bar{\varphi}_t^{(i)}a(y)+\varphi_t^{(i)}(y)a^*(x))+a^*(x)a(y)\Bigr)W(t,0)\Omega}\Biggr\rvert\\
+2\sum_{i<j}\ass{\bar{\alpha}_i(n) \alpha_j(n)}\Biggl(n\ass{\braket{\varphi^{(i)}}{\varphi^{(j)}}}^n\ass{\braket{\xi}{\bar{\varphi}_t^{(j)}\varphi_t^{(j)}}_{L^2(\mathds{R}^6)}}+ d_n\Biggl\lvert\braket{C^*(\sqrt{n}\varphi_0^{(j)})(\varphi^{(i)})^{\otimes_n}}{W^*(t,0)\\
\int\de{x}\ide{y}\bar{\xi}(x,y)\Bigl( \sqrt{n}(\bar{\varphi}_t^{(j)}a(y)+\varphi_t^{(j)}(y)a^*(x))+a^*(x)a(y)\Bigr)W(t,0)\Omega}\Biggr\rvert\Biggr)\\
\leq \Biggl(K_1e^{K_2\ass{t}}(\sqrt{n}+1)\Bigl(\sum_{i\in\mathds{N}}\ass{\alpha_i(n)}^2+d_n\sum_{i<j}\ass{\bar{\alpha}_i(n) \alpha_j(n)} \Bigr)+2 n\sum_{i<j}\ass{\bar{\alpha}_i(n) \alpha_j(n)}\\\ass{\braket{\varphi^{(i)}}{\varphi^{(j)}}}^n\biggr)\norm{\xi}_{L^2(\mathds{R}^6)}\; .
  \end{split}
\end{equation*}
By Riesz's lemma, keeping in mind that, for large $n$, $d_n\sim n^{1/4}$, we can write
\begin{equation*}\begin{split}
  \norm{T_\phi(t)-n\sum_{i\in\mathds{N}}\ass{\alpha_i(n)}^2\varphi_t^{(i)}\varphi_t^{(i)}}_{L^2(\mathds{R}^6)}\leq 2n\sum_{i<j}\ass{\bar{\alpha}_i(n)\alpha_j(n)}\ass{\braket{\varphi^{(i)}}{\varphi^{(j)}}}^n + K_2 e^{K_3\ass{t}}n^{3/4}\\\Bigl(\norm{(\alpha_i(n))_{i\in\mathds{N}}}_{l^2} +\sum_{i<j}\ass{\bar{\alpha}_i(n)\alpha_j(n)}\Bigr)\; .
\end{split}\end{equation*}
Dividing by $\Tr T_\Phi(t)=n$ we obtain the corresponding $L^2(\mathds{R}^6)$-bound for the integral kernel of the one-particle reduced density matrix, hence the bound in Hilbert-Schmidt norm. Since $\alpha_i(n)$ are uniformly bounded, we obtain:
\begin{equation}\label{eq:13}
  \begin{split}
    \norm{\Tr_1\rho_{\Phi}(t)-\sum_{i\in\mathds{N}}\ass{\alpha_i(n)}^2\ketbra{\varphi_t^{(i)}}{\varphi_t^{(i)}}}_{HS}\leq K_1\sum_{i<j}\ass{\bar{\alpha}_i(n)\alpha_j(n)}\ass{\braket{\varphi^{(i)}}{\varphi^{(j)}}}^n + K_2 e^{K_3\ass{t}}\frac{1}{n^{1/4}}\; .
  \end{split}
\end{equation}
Consider now $\norm{\Tr_1\rho_{\Phi}(t)-\norm{(\alpha_i)_{i\in\mathds{N}}}^{-2}_{l^2}\sum_{i\in\mathds{N}}\ass{\alpha_i}^2\ketbra{\varphi_t^{(i)}}{\varphi_t^{(i)}}}_{HS}$, we can write
\begin{equation*}
  \begin{split}
    \norm{\Tr_1\rho_{\Phi}(t)-\norm{(\alpha_i)_{i\in\mathds{N}}}^{-2}_{l^2}\sum_{i\in\mathds{N}}\ass{\alpha_i}^2\ketbra{\varphi_t^{(i)}}{\varphi_t^{(i)}}}_{HS}\leq \norm{\sum_{i\in\mathds{N}}\Bigl(\ass{\alpha_i(n)}^2-\norm{(\alpha_i)_{i\in\mathds{N}}}^{-2}_{l^2}\ass{\alpha_i}^2\Bigr)\ketbra{\varphi_t^{(i)}}{\varphi_t^{(i)}}}_{HS} \\+ \norm{\Tr_1\rho_{\Phi}(t)-\sum_{i\in\mathds{N}}\ass{\alpha_i(n)}^2\ketbra{\varphi_t^{(i)}}{\varphi_t^{(i)}}}_{HS}\; .
  \end{split}
\end{equation*}
The first term satisfies the following inequality:
\begin{equation*}
  \begin{split}
    \norm{\sum_{i\in\mathds{N}}\Bigl(\ass{\alpha_i(n)}^2-\norm{(\alpha_i)_{i\in\mathds{N}}}^{-2}_{l^2}\ass{\alpha_i}^2\Bigr)\ketbra{\varphi_t^{(i)}}{\varphi_t^{(i)}}}_{HS}\leq \sum_{i<j}\ass{\bar{\alpha}_i(n)\alpha_j(n)}\ass{\braket{\varphi^{(i)}}{\varphi^{(j)}}}^n\; .
  \end{split}
\end{equation*}

\begin{acknowledgments}
  The author would like to thank Professor Giorgio Velo, for having introduced him to the use of $\theta_{n,m}$ states for mean field limits, and all other interesting discussions.
\end{acknowledgments}
\bibliography{c-e_st_sup}
\end{document}